\documentclass[journal]{IEEEtran}
\usepackage{amsmath,amsfonts,amssymb,amsthm}
\usepackage{algorithmic}
\usepackage{algorithm}
\usepackage{array}
\usepackage[caption=false,font=normalsize,labelfont=sf,textfont=sf]{subfig}
\usepackage{cite}
\usepackage{textcomp}
\usepackage{url}
\usepackage{verbatim}
\usepackage{graphicx}
\usepackage{amssymb} 
\usepackage{cite}
\usepackage{xcolor}  
\usepackage{booktabs}
\usepackage[hidelinks]{hyperref}

\theoremstyle{plain}
\newtheorem{theorem}{Theorem}
\newtheorem{lemma}{Lemma}

\newtheorem{corollary}{Corollary}

\theoremstyle{definition}
\newtheorem{definition}{Definition}

\theoremstyle{remark}
\newtheorem{remark}{Remark}

\theoremstyle{definition}

\begin{document}
\title{Equilibrium-Free Contraction Stability Analysis for Grid-Forming Converter-Based Microgrids}
\author{Shijie Peng,~\IEEEmembership{Graduate Student Member, IEEE}, 
        Xiuqiang He,~\IEEEmembership{Member, IEEE}, 
        Xi Ru, 
        Hua Geng~\IEEEmembership{Fellow, IEEE}}

\maketitle
\begin{abstract}
Renewable-driven microgrids dominated by grid-forming (GFM) converters are subject to persistent power fluctuations, 
making equilibrium-known stability assessments restrictive.
This paper develops an equilibrium-free contraction stability method based on semi-contraction theory. 
By formulating the system in a symmetry-aware projected state space, 
the intrinsic rotational mode induced by uniform angle shifts is removed. 
A blockwise Jacobian decomposition is introduced to characterize the coupled active and reactive power dynamics, 
yielding a computable regional contraction condition.
This condition is then converted into forward-invariant stability certificates that provide trajectory-level performance guarantees.
For autonomous operation without disturbances, the method provides an equilibrium-free nonlinear stability characterization 
together with an estimation of the region of attraction (ROA). 
For non-autonomous operation under disturbances, it derives explicit bounds for quasi-steady tracking under slowly varying injections 
and for robustness under fast or composite disturbances.
Case studies on a 9-bus system validate the proposed method.
\end{abstract}


\vspace{-5pt}
\section{Introduction}
\vspace{-3pt}

\IEEEPARstart{S}{tability} is fundamental to the reliable operation of modern power grids.
As renewable penetration increases and synchronous inertia declines, 
grid-forming (GFM) converters are becoming essential for providing autonomous voltage and frequency support in islanded microgrids~\cite{lin2020roadmap}.
However, these microgrids exhibit nonlinearly coupled dynamics arising from network power flows and converter-control interactions.
Persistent renewable and load fluctuations can continuously excite these dynamics, causing the operating condition to evolve over time
rather than settle to a stationary post-disturbance point~\cite{Mohandes2019review, Anvari_2016}.
Under such conditions, stability assessments anchored to a fixed equilibrium become restrictive.
This raises a central question: 
how can stability be characterized for GFM converter-based microgrids
under persistent fluctuations without relying on a precomputed equilibrium?

Classical stability analysis is commonly organized around an equilibrium-known viewpoint,
where the system behavior is evaluated with respect to a known operating equilibrium. 
Representative tools include Lyapunov direct methods and energy-function approaches~\cite{Kundur1994, Chiang2011}.
These methods provide powerful means for assessing stability 
without exhaustive time-domain simulations. 
However, in power grids and microgrids subject to persistent renewable and load variations,
repeatedly updating the reference equilibrium and reconstructing the associated stability assessment 
can be computationally burdensome 
and may not yield a practical certificate~\cite{Milano2018}.
Thus, these classical analysis methods are not fully suited to stability assessment under such conditions in modern power grids.


These limitations become more pronounced in GFM converter-based islanded microgrids.
Without a stiff grid reference, system stability naturally depends on relative rather than absolute phase angles, 
meaning that selecting one converter as the internal reference may introduce redundant artificial constraints~\cite{Schiffer2014DroopStability, Florian2012}. 
Meanwhile, lossy networks with high $R/X$ ratios and converter-control interactions couple angle and voltage dynamics, 
making conventional $P$--$\theta$/$Q$--$V$ decoupling analysis insufficient~\cite{2020microgrid, Vorobev2018}.
Together, these features make the stability analysis of GFM converter-based islanded microgrids more challenging.

Although equilibrium-free perspectives on stability are not entirely new, 
their systematic extension to power systems remains nontrivial~\cite{Willems1974, Hill2006}. 
Concepts such as partial and set stability can characterize relative-angle behavior and invariant-set convergence~\cite{vorotnikov2005partial,bhatia2002stability}, 
yet they rarely provide computable criteria in converter-based microgrid applications.
Recent work on augmented synchronization links synchronization behavior to convergence toward equilibrium sets, 
offering insight beyond classical analysis~\cite{yang2024augment}. 
Relatedly, Kuramoto-like synchronization methods provide useful relative-angle insights 
without explicit equilibrium computation~\cite{dorfler2013synchronization}. 
However, these largely algebraic approaches often rely on simplified oscillator models, 
limiting their applicability to microgrids with coupled angle-voltage dynamics~\cite{simpson2013synchronization,Zhu2018}. 
These observations motivate an equilibrium-free stability analysis method that accounts for relative-angle behavior 
and angle-voltage coupling while yielding analytically computable stability certificates.

In practice, stability is often judged by whether frequency and voltage deviations 
remain bounded and return to acceptable operating regions, rather than by convergence 
to a fixed equilibrium~\cite{yang2024augment}. 
This pragmatic viewpoint calls for a method that can address both autonomous and non-autonomous operation.
For autonomous operation without exogenous disturbances, 
the method should characterize regional stability and synchronization without prior equilibrium computation.
For non-autonomous operation subjected to disturbances, 
it should quantify bounded trajectory deviations caused by 
persistent power variations and fast perturbations.
To this end, this paper develops an equilibrium-free stability method rooted in contraction theory~\cite{LOHMILLER1998, FBbook}.
By focusing on relative trajectory behavior in a projected space, 
the proposed method accommodates operating-point drift while capturing coupled angle-voltage converter dynamics. 
It further yields computable forward-invariant stability certificates that cover both spontaneous synchronization in autonomous systems 
and disturbance-induced deviations in non-autonomous systems.

The main contributions are summarized as follows:
\begin{enumerate}
    \item An equilibrium-free stability analysis method is developed for GFM converter-based microgrids, 
    where stability under renewable-driven operating-point drift is characterized by trajectory contraction 
    rather than convergence to a precomputed equilibrium.

    \item A computable contraction condition is derived in a projected transverse subspace, 
    where the neutral rotational mode is removed while the coupled angle-voltage dynamics are retained.

    \item Forward-invariant stability certificates are further constructed. 
    For autonomous operation, the resulting 
    invariant sets extend the classical ROA concept to an equilibrium-free setting. 
    For non-autonomous operation, the resulting state-space tubes provide explicit tracking and 
    robustness guarantees under slowly varying injections and fast perturbations.

\end{enumerate}

The remainder of this paper is organized as follows.
Section~II presents the preliminaries and problem formulation. 
Section~III develops the equilibrium-free contraction analysis method and the main stability condition.
Section~IV constructs stability assessment and dynamic performance certification for autonomous and non-autonomous operation.
Section~V validates the proposed method through case studies.
Finally, Section~VI concludes the paper.

\textit{Notation:}
$\mathbb{R}$ and $\mathbb{R}_{+}$ denote the sets of real and positive real numbers, respectively. 
$\mathbf{1}_n$ denotes the $n$-dimensional all-one vector. 
$\operatorname{diag}(x)$ denotes the diagonal matrix with entries $x$, and $\operatorname{blkdiag}(A,B)$ denotes the block-diagonal matrix with diagonal blocks $A$ and $B$. 
For a symmetric matrix $A$, $\lambda_{\min}(A)$ and $\lambda_{\max}(A)$ denote its minimum and maximum eigenvalues, respectively. 
The notation $\|\cdot\|_2$ denotes the Euclidean norm for vectors and the induced $2$-norm for matrices.

\vspace{-6pt}
\section{Preliminaries and Formulations}
\vspace{-1pt}

\subsection{System Modeling}

We consider a microgrid modeled as a graph $\mathcal{G}=(\mathcal{V},\mathcal{E})$ 
after Kron reduction of passive loads~\cite{dorfler2012kron}, where 
$\mathcal{V} = \{1,2,\dots,N\}$ denotes the set of converter-interfaced buses and
$\mathcal{E} \subseteq \mathcal{V} \times \mathcal{V}$ denotes the set of effective electrical interconnections.
A schematic illustration is shown in Fig.~\ref{fig:sys_1}.
The network admittance matrix is denoted by $Y = G + jB \in \mathbb{C}^{N\times N}$,
where $G$ and $B$ are the conductance and susceptance matrices, respectively.
For the voltage phasor $V_i \angle \theta_i$ at bus $i \in \mathcal{V}$,
the active and reactive power injections are given by the standard power-flow equations
\begin{subequations} \label{eq:power-flow}
\begin{align}
    P_i &= V_i^2 G_{ii} + V_i \sum_{k \neq i} V_k \left( G_{ik}\cos\theta_{ik} + B_{ik}\sin\theta_{ik} \right) \\
    Q_i &= -V_i^2 B_{ii} + V_i \sum_{k \neq i} V_k \left( G_{ik}\sin\theta_{ik} - B_{ik}\cos\theta_{ik} \right),
\end{align}
\end{subequations}
where $\theta_{ik} := \theta_i - \theta_k$.
These expressions retain the effects of line resistance and are therefore applicable to general lossy networks,
which is essential for converter-based systems with non-negligible $R/X$ ratios.
\begin{figure}[t]
\centering
\includegraphics[width=0.77\linewidth]{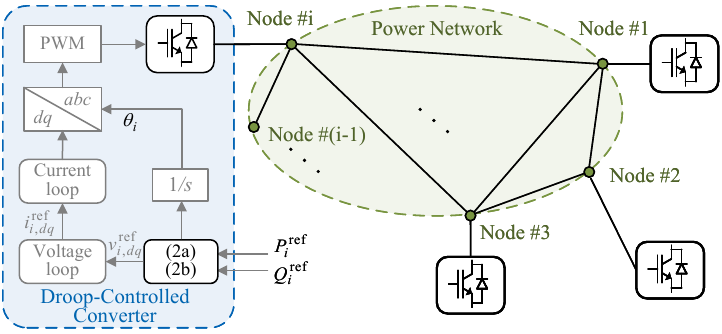}
\vspace{-6pt} 
\caption{Schematic of the droop-controlled islanded microgrid.}
\label{fig:sys_1}
\vspace{-14pt} 
\end{figure}

Among GFM converter-based microgrids, this paper focuses on the droop-controlled setting, 
which is characterized by low or negligible physical inertia, pronounced sensitivity to disturbances, 
time-varying operating conditions, and strong angle-voltage coupling~\cite{2020microgrid, Schiffer2014DroopStability, wang2023droop}. 
Specifically, the coupling here denotes the nonlinear dependence of network power-flow equations on both phase-angle differences and voltage magnitudes, 
while these coupled power terms drive the nodal dynamics through additive input channels in the control law.
Under the standard time-scale separation assumption~\cite{He2021}, 
the inner voltage and current loops are neglected because of their much faster responses, 
leaving the state of each converter node defined as $x_i=[\theta_i,V_i]^\top$,
where $\theta_i$ and $V_i$ denote the phase angle and voltage magnitude, respectively.
The corresponding dynamics are governed by
\begin{subequations}\label{eq:droop_sys1}
\begin{align}
    \dot{\theta}_i &= \omega_{\mathrm{nom}} - m_{p,i} \big( P_i - P^{\mathrm{ref}}_i \big) \\
    \tau_{v,i} \dot{V}_i &= V_{\mathrm{nom},i} - V_i - n_{q,i} \big( Q_i - Q^{\mathrm{ref}}_i \big),
\end{align}
\end{subequations}
where $P_i^{\mathrm{ref}}$ and $Q_i^{\mathrm{ref}}$ are the active and reactive power references,
$\omega_{\mathrm{nom}}$ and $V_{\mathrm{nom},i}$ are the nominal frequency and voltage setpoints,
$m_{p,i}>0$ and $n_{q,i}\ge 0$ are the droop coefficients, and $\tau_{v,i}>0$ is the time constant of the voltage control loop.
Stacking the nodal states yields $x=[\theta^\top,V^\top]^\top \in \mathbb{R}^{2N}$,
so that the autonomous system can be written compactly as $\dot{x}=f(x)$.

\vspace{-10pt} 
\subsection{Operating Conditions and Admissible Domain}
The autonomous operation corresponds to constant power references.
To capture the volatility and stochasticity of renewable generation and demand, 
we model the references as
time-varying perturbations around nominal values~\cite{Ru2024Slowly}:
\[
P_i^{\mathrm{ref}}(t)=P_{i,0}^{\mathrm{ref}}+ p_{i}(t),\quad
Q_i^{\mathrm{ref}}(t)=Q_{i,0}^{\mathrm{ref}}+ q_{i}(t),
\]
where $p_{i}(t)$ and $q_{i}(t)$ are time-varying exogenous inputs.
The resulting non-autonomous system can be written as
\begin{equation}\label{eq:droop_sys}
\dot{x}=f(x,u(t)),
\end{equation}
where $u(t):=[p^\top(t), q^\top(t)]^\top$. 
For brevity, write the system as $\dot{x}=f(x,t)$.
In the sequel, the autonomous case characterizes spontaneous synchronization and regional stability, 
whereas the non-autonomous case quantifies robust performance under persistent variations and disturbances.

To formulate meaningful stability conditions, 
we restrict the analysis to a physically admissible operating domain $\mathcal D$.
Specifically, let $\mathcal D$ be a compact set such that
$V_i \in [\underline{V}, \overline{V}]$ for all $i \in \mathcal{V}$,
and $|\theta_{ik}| \le \gamma_{\max} < \pi/2$ for all $(i,k)\in\mathcal E$.
Within $\mathcal D$, the trigonometric terms satisfy
\begin{equation} \label{eq:trig_bounds}
    \cos\theta_{ik} \ge \cos\gamma_{\max} > 0, \quad |\sin\theta_{ik}| \le \sin\gamma_{\max},
\end{equation}
and the voltage products satisfy
\begin{equation} \label{eq:voltage_bounds}
    \underline{V}^2 \le V_i V_k \le \overline{V}^2 .
\end{equation}

These bounds yield the uniform estimates essential for the subsequent blockwise contraction stability analysis.

\vspace{-10pt} 
\subsection{Semi-Contraction Preliminaries}

Contraction theory offers an equilibrium-free view of stability by studying convergence 
between neighboring trajectories rather than toward a fixed equilibrium. 
Such convergence can be certified through the matrix measure of the system Jacobian along trajectories, 
making the approach suitable for drifting operating conditions. 
Semi-contraction further extends this idea by using a seminorm to exclude physically neutral directions from the stability characterization.

We next recall the semi-contraction notions used in this paper. 
Standard definitions of seminorms and induced matrix measures follow~\cite[Defs.~5.1 and~5.4]{FBbook}.

\begin{definition}[{Semi-contraction~\cite[Def.~5.9]{FBbook}}] \label{def:semi_contraction}
\textit{Consider a time-varying nonlinear system}
\begin{equation} \label{eq:nonaut_sys}
\dot{x}=f(x,t), \qquad x\in\mathbb{R}^n,\ t\in\mathbb{R}_{+},
\end{equation}
\textit{where $f(\cdot,t)$ is continuously differentiable in $x$, and $Df(x,t)$ denotes its Jacobian. 
Let $\|\cdot\|_s$ be a seminorm with induced matrix measure $\mu_{s}$, and let $\mathcal{D}\subseteq\mathbb{R}^n$ be convex. 
The system is semi-contracting on $\mathcal{D}$ with rate $c>0$ if}
\begin{equation}
\mu_{s}\big(Df(x,t)\big)\le -c,
\qquad \forall x\in\mathcal{D},\ \forall t\ge0.
\end{equation}
\end{definition}

\begin{remark}
In islanded microgrids, the seminorm removes the neutral rotational direction, 
allowing stability to be characterized on the physically meaningful transverse subspace.
\end{remark}

Fig.~\ref{fig:con_1} illustrates the projected convergence behavior. 
Since kernel directions are ignored in Definition~\ref{def:semi_contraction}, 
their infinitesimal invariance is required for a well-posed projected analysis.
\begin{figure}[t]
\centering
\includegraphics[width=0.72\linewidth]{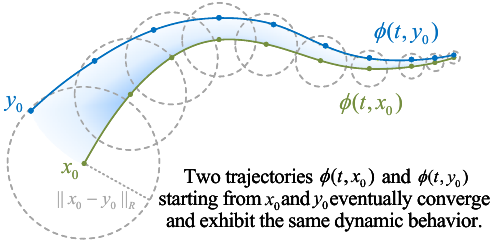}
\vspace{-6pt} 
\caption{Semi-contraction and convergence in the $R$-projected space. 
The gray dashed balls indicate a distance bound that decays over time.}
\label{fig:con_1}
\vspace{-15pt} 
\end{figure}



\begin{definition}[{Infinitesimal kernel invariance~\cite[Def.~12]{Jafarpour2022}}]
\label{def:kernel_invariance}
\textit{Consider system~\eqref{eq:nonaut_sys} and a seminorm $\|\cdot\|_s$ with kernel 
$\ker\|\cdot\|_s := \{ v\in\mathbb{R}^n \mid \|v\|_s = 0 \}$.
The kernel $\ker\|\cdot\|_s$ is said to be infinitesimally invariant on $\mathcal D$ if}
\begin{equation}
Df(x,t)\,\ker\|\cdot\|_s \subseteq \ker\|\cdot\|_s,
\qquad \forall x\in\mathcal D,\ \forall t\ge 0.
\end{equation}
\end{definition}

\begin{remark}
This condition ensures that kernel perturbations remain dynamically decoupled from the projected dynamics.
\end{remark}

We next develop the projected stability analysis and derive the corresponding contraction condition.
\section{Contraction-Based Stability Analysis}
This section establishes a contraction stability condition for the system.
By constructing a projected representation to handle rotational symmetry and verifying kernel invariance, 
we analyze the projected Jacobian blockwise to quantify the angle, voltage, and cross-coupling effects. 
This yields a computable condition that 
underpins the forward-invariant set and robustness results in Section~IV.

\vspace{-10pt}
\subsection{Projected Dynamics and Kernel Invariance}
Because power injections are invariant to uniform shifts in phase angles $\theta \mapsto \theta + \alpha \mathbf{1}_N$, 
the system dynamics possess a neutral rotational mode along $\mathbf{1}_N$.
To remove this mode, 
let $R_\theta \in \mathbb{R}^{(N-1)\times N}$ have orthonormal rows spanning $\mathrm{span}\{\mathbf{1}_N\}^\perp$. 
Then $R_\theta \mathbf{1}_N = \mathbf{0}$ and $R_\theta R_\theta^\top = I_{N-1}$. 
Since the rotational symmetry acts only on the angular states, the voltage states are left unchanged. 
Accordingly, define the projection matrix $R \in \mathbb{R}^{(2N-1)\times 2N}$ as
\begin{equation}\label{eq:projection_matrix}
R :=
\begin{bmatrix}
R_\theta & \mathbf{0} \\
\mathbf{0} & I_N
\end{bmatrix},
\qquad
\|x\|_R := \|Rx\|_2 .
\end{equation}
The seminorm $\|x\|_R$ thus captures only deviations transverse to the rotational symmetry. 
Since $R$ has orthonormal rows, its Moore--Penrose pseudoinverse reduces to $R^\dagger = R^\top$, 
yielding the orthogonal projector
\begin{equation}
\Pi := R^\dagger R
= \mathrm{blkdiag}(R_\theta^\top R_\theta,\, I_N).
\end{equation}

As $\Pi$ projects the state onto the transverse subspace, 
$\|x\|_R = \|\Pi x\|_2$ measures the Euclidean norm of the projected state.
By construction, 
$\ker\|\cdot\|_R=\ker(R)=\mathrm{span}\{\mathbf{1}_{2N}^\theta\}$,
where $\mathbf{1}_{2N}^\theta := [\mathbf{1}_N^\top,\, \mathbf{0}_N^\top]^\top$.
The time-varying power references enter additively and therefore do not affect the Jacobian $J(x,t):=Df(x,t)$. 
Rotational symmetry then implies
$J(x,t)\mathbf{1}_{2N}^\theta=\mathbf{0}_{2N}$ 
for all $x\in\mathcal D$ and $t\ge0$. 
Equivalently,
\[
    J(x,t)\ker(R)\subseteq \ker(R),
\qquad
\forall x\in\mathcal D,\ \forall t\ge 0.
\]

Thus, $\ker\|\cdot\|_R$ is infinitesimally invariant on $\mathcal D$, 
so the corresponding $R$-weighted matrix measure $\mu_R(\cdot)$ is well defined for the subsequent semi-contraction analysis.

\vspace{-5pt}
\subsection{Blockwise Jacobian Analysis for Semi-Contraction} \label{subsec:semi_contraction}

For the nonlinear system $\dot{x}=f(x,t)$, the Jacobian $J(x,t)$ describes 
the variational dynamics between neighboring trajectories, 
i.e., $\dot{\delta x}=J(x,t)\delta x$. 
Here, $J(x,t)$ is used not as a small-signal linearization around an equilibrium, 
but as a differential description of how trajectory deviations evolve, 
thereby yielding a stability certificate for the original system.

With kernel invariance established,  
semi-contraction can be assessed from the projected Jacobian
$J_R(x,t)=RJ(x,t)R^\top$ 
as the matrix measure satisfies $\mu_R(J(x,t)) = \mu_2(J_R(x,t))$~\cite[Thm.~6]{Jafarpour2022}.
Partitioning $J(x,t)$ into angle and voltage blocks, the projected Jacobian
$J_R(x,t)\in\mathbb{R}^{(2N-1)\times(2N-1)}$ is formulated as
\begin{equation} \label{eq:projected_jacobian}
    J_{R}(x,t) = 
    \begin{bmatrix}
        R_\theta J_{\theta\theta} R_\theta^\top & R_\theta J_{\theta V} \\
        J_{V\theta} R_\theta^\top & J_{VV}
    \end{bmatrix}.
\end{equation}

Semi-contraction of the projected dynamics is governed by the symmetric part of $J_R(x,t)$,
defined as $\mathcal{S}(x,t):=\frac{1}{2}(J_R+J_R^\top)$, since the Euclidean matrix measure satisfies 
$\mu_2(J_R)=\lambda_{\max}(\mathcal{S})$.
It admits the block form
\begin{equation} \label{eq:symmetric_S}
    \mathcal{S}(x,t) =
\begin{bmatrix}
        \mathcal{S}_{\theta\theta} & \mathcal{S}_{\theta V} \\
        \mathcal{S}_{\theta V}^\top & \mathcal{S}_{VV}
\end{bmatrix}.
\end{equation}
Here, $\mathcal{S}_{\theta\theta}$ and $\mathcal{S}_{VV}$ represent the angle and voltage channels, respectively, 
while $\mathcal{S}_{\theta V}$ captures their cross-coupling.

This decomposition maps the projected Jacobian to the underlying physical mechanisms, 
with details provided in Appendix~\ref{app:jacobian_derivation}. 
The semi-contraction analysis thus reduces to establishing uniform blockwise bounds over $\mathcal{D}$.

\subsubsection{Active Power--Angle Block Characterization}
The block $\mathcal S_{\theta\theta}$ characterizes the synchronizing action in the active power--angle channel.
Although $J_{\theta\theta}(x,t)$ has zero row sums, 
it is generally asymmetric in lossy networks with heterogeneous active-power control coefficients.
Its symmetric part admits the decomposition
\begin{equation}
    \frac{1}{2}(J_{\theta\theta} + J_{\theta\theta}^\top) = -L_{\mathrm{sym}}(x) + \Delta(x),
    \label{eq:laplacian_decomposition}
\end{equation}
where $L_{\mathrm{sym}}(x)$ is a symmetric weighted Laplacian associated with the effective synchronizing weights $\{w_{ik}^{\mathrm{sym}}(x)\}$, and
$\Delta(x):=\operatorname{diag}\!\left(\frac{1}{2}(J_{\theta\theta}+J_{\theta\theta}^\top)\mathbf 1_N\right)$
collects the residual diagonal term due to the heterogeneity.
Over the admissible domain $\mathcal D$,~\eqref{eq:trig_bounds}--\eqref{eq:voltage_bounds} imply
$w_{ik}^{\mathrm{sym}}(x)\ge \underline w_{ik}$, where
\begin{multline}
    \underline{w}_{ik} := \frac{1}{2} \underline{V}^2 \Big( (m_{p,i} + m_{p,k}) B_{ik} \cos\gamma_{\max} \\
    - |m_{p,i} - m_{p,k}| |G_{ik}| \sin\gamma_{\max} \Big).
\end{multline}
Let $\underline L$ denote the constant Laplacian constructed from $\{\underline w_{ik}\}$. 
Then $L_{\mathrm{sym}}(x)-\underline L$ is a Laplacian with nonnegative edge weights and hence positive semidefinite.
Define the heterogeneity penalty
\begin{equation}
\delta_\theta := \sup_{x \in \mathcal{D}} \max_{i \in \mathcal{V}} \Delta_{ii}(x).
\end{equation}

\begin{lemma}[Uniform Synchronizing Bound]
\label{lem:bound_theta}
For the projected active power--angle block $\mathcal{S}_{\theta\theta}(x,t)$, it holds that
\begin{equation}
\lambda_{\max}\bigl(\mathcal{S}_{\theta\theta}(x,t)\bigr) \le -c_\theta < 0, \quad \forall x\in\mathcal D,\ \forall t\ge 0,
\label{eq:htheta_bound}
\end{equation}
provided that $\lambda_2(\underline{L}) > \delta_\theta$. 
Here, the uniform synchronizing margin is defined as $c_\theta := \lambda_2(\underline{L}) - \delta_\theta$, 
where $\lambda_2(\underline{L})$ is the algebraic connectivity of $\underline{L}$.
\end{lemma}

Lemma~\ref{lem:bound_theta}, proved in Appendix~\ref{appendix:proof_theta_block}, 
indicates that the angle channel is uniformly contractive when network synchrony outweighs 
the degradation due to the heterogeneity.
Specifically, $\lambda_2(\underline{L})$ measures the network synchronizing strength, 
primarily shaped by susceptances, voltage levels, and electrical connectivity, 
whereas $\delta_\theta$ captures the degradation induced by network losses and nonuniform droop gains.
In the ideal limit of a purely inductive and homogeneous network,
$\delta_\theta=0$ and the block reduces to an ideal Laplacian synchronizing structure.
\begin{figure}[t]
    \centering
    \includegraphics[width=0.78\linewidth]{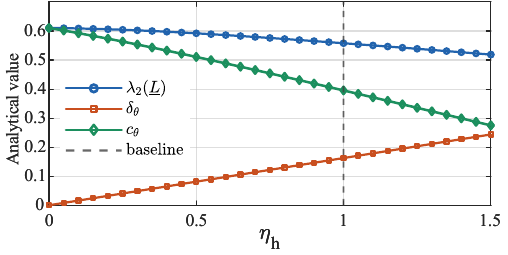}
    \vspace{-9pt}
    \caption{Impact of heterogeneity on $c_\theta$.
    The droop gains are parameterized as $m_p(\eta_h)=\bar{m}_p \mathbf{1}+\eta_h(m_p^0-\bar{m}_p \mathbf{1})$,
    where $\bar{m}_p:=\frac{1}{N}\mathbf{1}^\top m_p^0$ is the mean droop gain.
    The baseline case is $\eta_h=1$ (dashed line), and $\eta_h=0$ gives the uniform-droop case.}
    \label{fig:stage3_ctheta_sensitivity}
    \vspace{-16pt}
\end{figure}

Fig.~\ref{fig:stage3_ctheta_sensitivity} illustrates this engineering interpretation by 
decomposing $c_\theta$ under a mean-preserving increase in droop heterogeneity.
As the heterogeneity increases, the penalty $\delta_\theta$ rises markedly, while $\lambda_2(\underline{L})$ reduces mildly, 
thereby eroding the net margin $c_\theta$.
This aligns with the intuition that homogeneous settings promote coherent synchronization.

\subsubsection{Reactive Power--Voltage Block Characterization}
The block $\mathcal S_{VV}$ characterizes restoring in the reactive power--voltage channel.
Unlike the active power--angle block, it does not admit a useful Laplacian decomposition.
Its spectral properties are instead governed by the local voltage-control dynamics and network reactive coupling.
Accordingly, we bound its spectrum via Gershgorin's theorem~\cite[Thm.~6.1.1]{horn2012matrix}.

For each node $i\in\mathcal V$, define the Gershgorin disc center and radius of $\mathcal S_{VV}(x,t)$ as
\begin{equation}
    c_i(x):=[\mathcal S_{VV}(x,t)]_{ii},
    \quad
    r_i(x):=\sum_{k\neq i}\big|[\mathcal S_{VV}(x,t)]_{ik}\big|.
\end{equation}

Over the compact domain $\mathcal D$,~\eqref{eq:trig_bounds}--\eqref{eq:voltage_bounds} yield uniform bounds
$c_i(x)\le \bar c_i$ and $r_i(x)\le \bar r_i$, where
\begin{equation}
    \bar{c}_i := \frac{2n_{q,i}\underline{V} B_{ii}-1}{\tau_{v,i}} + \sum_{k \neq i}\frac{n_{q,i}}{\tau_{v,i}} \max_{|\theta| \le \gamma_{\max}} \big|B_{ik} \cos\theta - G_{ik} \sin\theta\big|,
\end{equation}
and the worst-case radius $\bar{r}_i$ is given by
\begin{equation}
    \bar{r}_i := \frac{1}{2} \sum_{k \neq i} \left( \frac{n_{q,i}}{\tau_{v,i}} + \frac{n_{q,k}}{\tau_{v,k}} \right) \overline{V} \max_{|\theta| \le \gamma_{\max}} \big|B_{ik} \cos\theta - G_{ik} \sin\theta\big|.
\end{equation}
\begin{lemma}[Uniform Voltage Restoring Bound]
\label{lem:bound_V}
For the projected reactive power--voltage block $\mathcal{S}_{VV}(x,t)$, it holds that
\begin{equation}
        \lambda_{\max}\bigl(\mathcal S_{VV}(x,t)\bigr)\le -c_V<0,
    \quad \forall x\in\mathcal D,\ \forall t\ge 0,
\end{equation}
provided that $\bar{c}_i + \bar{r}_i < 0$ for all $i \in \mathcal{V}$. 
The uniform voltage restoring margin is defined as
$c_V := -\max_{i \in \mathcal{V}} (\bar{c}_i + \bar{r}_i)$.
\end{lemma}

Lemma~\ref{lem:bound_V}, proved in Appendix~\ref{appendix:proof_voltage_block}, 
determines whether the local voltage-restoring action at each bus dominates the worst-case reactive coupling from the network.
Here, $\bar c_i$ captures the local restoring term, mainly shaped by the voltage filter time constant and self-admittance, 
while $\bar r_i$ bounds the aggregate coupling from neighboring buses.
Thus, $\bar c_i+\bar r_i<0$ for all $i\in\mathcal V$ ensures uniform contractivity over $\mathcal D$, 
and the resulting margin $c_V$ provides a computable certificate linking controller tuning and network parameters to voltage restoring performance.
Notably, enhanced restoring is facilitated by reduced filter time constants and appropriately scaled reactive droop gains.

\subsubsection{Cross-Coupling Block Characterization}
The block $\mathcal S_{\theta V}$ captures the interaction between the active power--angle and reactive power--voltage channels, 
which is exacerbated by network conductances and heterogeneous operating conditions.
Unlike the two diagonal blocks, it provides no intrinsic restoring action.
Its effect is therefore quantified by a uniform norm bound over $\mathcal D$.
\begin{lemma}[Uniform Cross-Coupling Bound] \label{lem:bound_cross}
For the projected off-diagonal block $\mathcal{S}_{\theta V}(x,t)$, 
there exists a constant $\beta \ge 0$ such that
\begin{equation}
    \|\mathcal S_{\theta V}(x,t)\|_2 \le \beta,
    \qquad \forall x\in\mathcal D,\ \forall t\ge 0, 
\end{equation}
where the uniform cross-coupling margin $\beta$ is given by the constrained maximum 
$ \beta := \max_{x \in \mathcal{D}} \|\mathcal{S}_{\theta V}(x,t)\|_2 $.
\end{lemma}

The margin $\beta$ quantifies the worst-case bidirectional coupling between 
the active and reactive control loops over $\mathcal D$.
A larger $\beta$ signifies stronger dynamic interaction and hence greater erosion of 
the synchronizing and restoring effects provided by the two diagonal blocks.
Consequently, overall contraction stability requires the combined diagonal margins to dominate this cross-coupling.

\vspace{-10pt}
\subsection{Contraction Stability Conditions}
The uniform block bounds above provide a computable condition for semi-contraction of the projected dynamics. 
Assume that, for all $x\in\mathcal D$,
\begin{equation} \label{eq:matrix_bounds}
\left\{
\begin{aligned}
    &\, \lambda_{\max} (\mathcal{S}_{\theta\theta})\le -c_\theta, \\
    & \, \lambda_{\max} (\mathcal{S}_{VV}) \le -c_V, \\
    &\, \|\mathcal{S}_{\theta V}\|_2 \le \beta,
\end{aligned}
\right.
\end{equation}
Then the following theorem gives the resulting regional semi-contraction condition.

\begin{theorem}[Regional Semi-Contraction Condition]\label{thm:semi_contraction}
Consider the droop-controlled islanded system $\dot x = f(x,t)$ in~\eqref{eq:droop_sys}, 
defined on a convex admissible domain $\mathcal{D}$. 
Suppose that the projected Jacobian satisfies~\eqref{eq:matrix_bounds} 
for all $x \in \mathcal{D}$.
If
\begin{equation}\label{eq:contraction_condition}
    c_\theta c_V>\beta^2,
\end{equation}
then the system is semi-contracting on $\mathcal D$ with respect to the seminorm $\|\cdot\|_R$.
Specifically,
\begin{equation}\label{eq:contraction_rate}
    \mu_R(J(x,t))\le -c,
    \qquad \forall x\in\mathcal D,\ \forall t\ge 0,
\end{equation}
where the uniform semi-contraction rate is
\begin{equation}
    c := \frac{1}{2}\left(c_\theta + c_V - \sqrt{(c_\theta - c_V)^2 + 4\beta^2}\right) > 0.
\end{equation}
\end{theorem}

\begin{proof}
Let $z=[z_\theta^\top,z_V^\top]^\top\neq 0$, with $z_\theta\in\mathbb R^{N-1}$ and $z_V\in\mathbb R^N$.
By the block structure of $\mathcal S(x,t)$ and~\eqref{eq:matrix_bounds},
\begin{equation*}
\begin{aligned}
z^\top\mathcal S z
&= z_\theta^\top \mathcal S_{\theta\theta} z_\theta
 + z_V^\top \mathcal S_{VV} z_V
 + 2 z_\theta^\top \mathcal S_{\theta V} z_V \\
&\le -c_\theta\|z_\theta\|_2^2
     -c_V\|z_V\|_2^2
     +2\beta\|z_\theta\|_2\|z_V\|_2.
\end{aligned}
\end{equation*}

Define $u:=[\|z_\theta\|_2,\|z_V\|_2]^\top$. Then $u^\top u=z^\top z$ and
\[
z^\top\mathcal S z
\le
u^\top
\underbrace{
\begin{bmatrix}
-c_\theta & \beta\\
\beta & -c_V
\end{bmatrix}}_{=:M_c}
u.
\]
Hence, for every $z\neq 0$,
\[ \frac{z^\top \mathcal S z}{z^\top z} \le
\frac{u^\top M_c u}{u^\top u} \le  \lambda_{\max}(M_c). \]
Taking the supremum yields $\lambda_{\max}(\mathcal S(x,t)) \le \lambda_{\max}(M_c)$.

Since $M_c$ is symmetric,
\[ \lambda_{\max}(M_c) = \frac{1}{2} \left( -(c_\theta + c_V) + \sqrt{(c_\theta - c_V)^2 + 4\beta^2} \right) = -c. \]
Under~\eqref{eq:contraction_condition}, we have $c>0$. 
Recalling that \(\mu_R(J(x,t))=\lambda_{\max}(\mathcal S(x,t))\), 
it follows that, for all \(x\in\mathcal D\) and \(t\ge 0\),
\begin{equation*}
     \mu_R(J(x,t)) \le \lambda_{\max}(M_c) = -c < 0.  \qedhere 
\end{equation*}
\end{proof}

Theorem~\ref{thm:semi_contraction} establishes a computable condition on $\mathcal D$, 
from which trajectory-level stability properties follow directly.
These properties, however, hold only when trajectories remain in a forward-invariant subset of $\mathcal D$.
We defer the explicit construction of such sets to Section~IV 
and first state the following corollary based on this invariance.

\begin{corollary}[Trajectory Convergence under Contraction Stability] \label{cor:trajectory}
Suppose Theorem~\ref{thm:semi_contraction} holds with respect to the seminorm $\|\cdot\|_R$.
Let $\Omega \subseteq \mathcal{D}$ be a forward-invariant set. 
Then, as illustrated in Fig.~\ref{fig:con_1}, 
the following properties hold for any trajectories $\phi(t,x_0)$ and $\phi(t,y_0)$ evolving in $\Omega$ and all $t \ge 0$
\begin{itemize}
    \item[\textbf{(i)}] \textbf{Exponential decay of trajectory distances:}
    \begin{equation} \label{eq:incremental_stability}
        \|\phi(t,x_0)-\phi(t,y_0)\|_R \le e^{-ct}\|x_0-y_0\|_R .
    \end{equation}

    \item[\textbf{(ii)}] \textbf{Exponential convergence to the stable manifold:}
    For autonomous systems admitting an equilibrium $x^* \in \Omega$, all trajectories converge exponentially 
    to the affine manifold $\mathcal{M} := \{x^* + \ker(R)\}$ with rate $c$, satisfying
    \begin{equation} \label{eq:manifold_convergence}
        \inf_{v \in \ker(R)} \|\phi(t,x_0)-(x^*+v)\|_R \le e^{-ct}\|x_0-x^*\|_R .
    \end{equation}
\end{itemize}
\end{corollary}

\begin{proof}
Because $\Omega\subseteq\mathcal D$ is forward-invariant, 
Theorem~\ref{thm:semi_contraction} directly gives~\eqref{eq:incremental_stability} on $\Omega$.
For~\eqref{eq:manifold_convergence}, set $y_0=x^*+v$ with $v\in\ker(R)$ in~\eqref{eq:incremental_stability}. 
By rotational symmetry, the left-hand side becomes $\|\phi(t,x_0)-(x^*+v)\|_R$, while
$ \|x_0-(x^*+v)\|_R=\|x_0-x^*\|_R $
since $\|\cdot\|_R$ vanishes on $\ker(R)$.
Taking the infimum over $v\in\ker(R)$ gives~\eqref{eq:manifold_convergence}.
\end{proof}

Corollary~\ref{cor:trajectory} formalizes the trajectory interpretation shown in Fig.~\ref{fig:con_1}: 
within a forward-invariant set, projected distances between trajectories decay exponentially, 
implying convergence of autonomous operations toward the stable manifold.
Section~IV completes the analysis by constructing such invariant sets, 
thereby turning the condition into a stability certificate for both autonomous and non-autonomous operation.

\begin{remark}
Table~\ref{tab:comparison_energy_contraction} summarizes the comparison between classical Lyapunov direct methods and the proposed method.
\end{remark}


\begin{table}[t]
\vspace{-6pt}
\centering
\caption{Conceptual comparison with classical methods.}
\vspace{-3pt}
\label{tab:comparison_energy_contraction}
\begin{tabular}{c c c}
\toprule
\midrule
\textbf{Aspect} &  \textbf{Lyapunov direct method} & \textbf{Proposed} \\
\midrule
Object & State-to-equilibrium & Trajectory-to-trajectory \\
Equilibrium & Precomputed & Not prescribed \\
Symmetry & Reference fixing & $R$-projection \\
Time variation & Updated certificates & Unified framework \\
Linear analogy & Local linearization & Trajectory Jacobian \\
\bottomrule
\vspace{-20pt}
\end{tabular}
\end{table}


\section{Stability Assessment and Dynamic Performance Certification}
Theorem~\ref{thm:semi_contraction} provides a contraction-based condition on $\mathcal D$, 
whose practical use requires a forward-invariant set where the condition remains valid. 
This section converts the condition into computable stability certificates through invariant tubes in the projected state space. 
The method is then specialized to autonomous operation for equilibrium-free nonlinear stability assessment and 
to non-autonomous operation for tracking and robustness bounds under time-varying injections.

\vspace{-10pt}
\subsection{General Invariant-Tube Certificate}
We construct an invariant tube around an arbitrary reference trajectory, 
with its mismatch quantified by a residual term.

\begin{theorem}[Forward-Invariant Tube Under Semi-Contraction] \label{thm:invariant_tube}
Consider $\dot x=f(x,t)$, semi-contracting on a convex domain $\mathcal D$ under $\|\cdot\|_R$ with rate $c>0$.
Let $x_c(t)\in\mathcal D$ be an absolutely continuous reference trajectory.
Define the residual vector $d(t) := f(x_c(t),t) - \dot{x}_c(t)$ and its seminorm $r(t) := \|d(t)\|_R$. 
Let the scalar $\rho(t)$ satisfy
$\dot\rho(t)=-c\rho(t)+r(t)$ with 
$\rho(0)\ge \inf_{v\in\ker(R)}\|x(0)-(x_c(0)+v)\|_R$.
Then
\begin{equation}
    \inf_{v\in\ker(R)} \|x(t) - (x_c(t)+v)\|_R \le \rho(t), \quad \forall t \ge 0.
\end{equation}
Consequently, the time-varying tube
\begin{equation}
    \mathcal{T}(t) := \Bigl\{ x \;\Big|\; \inf_{v \in \ker(R)} \|x - (x_c(t) + v)\|_R \le \rho(t) \Bigr\}
\end{equation}
is forward-invariant.
\end{theorem}

\begin{proof}
By the mean-value theorem~\cite[Appx.~A]{khalil2002nonlinear},
\begin{equation*}
\begin{aligned}
    &f(x(t),t)-f(x_c(t),t) \\
    &\quad = \Bigl(\int_0^1 Df\bigl(x_c(t)+\alpha(x(t)-x_c(t)),t\bigr)\,d\alpha\Bigr)\bigl(x(t)-x_c(t)\bigr),
\end{aligned}
\end{equation*}
where $x_c(t)+\alpha(x(t)-x_c(t))\in\mathcal D$ for all $\alpha\in[0,1]$. 
Thus, semi-contraction on $\mathcal D$ gives
\begin{equation*}
\begin{aligned}
    D^+ \|x(t) - x_c(t)\|_R &\le -c\|x(t) - x_c(t)\|_R + \|d(t)\|_R \\
     & = -c\|x(t) - x_c(t)\|_R + r(t). 
\end{aligned}
\end{equation*}
where $D^+$ denotes the upper right-hand derivative~\cite[Appx.~C.2]{khalil2002nonlinear}.
With $\rho(0)\ge \inf_{v\in\ker(R)}\|x(0)-(x_c(0)+v)\|_R$, the comparison lemma~\cite[Lem.~3.4]{khalil2002nonlinear} yields
\[
\|x(t)-x_c(t)\|_R\le \rho(t), \quad \forall t\ge 0.
\]
Since $\|\cdot\|_R$ vanishes on $\ker(R)$, taking the infimum yields
\[
\inf_{v\in\ker(R)}\|x(t)-(x_c(t)+v)\|_R\le \rho(t), \quad \forall t\ge 0.
\]
Hence, $\mathcal T(t)$ is forward-invariant.
\end{proof}

Theorem~\ref{thm:invariant_tube} constructs an invariant tube around $x_c(t)+\ker(R)$.
The reference $x_c(t)$ does not need to satisfy the exact dynamics, as is often the case with an ideal power dispatch command. 
The residual $d(t)$ quantifies this mismatch and is propagated through 
the scalar comparison system $\dot\rho=-c\rho+r(t)$ to bound projected deviations.
We next specialize this certificate to autonomous and non-autonomous regimes.

\vspace{-10pt}
\subsection{Autonomous Operation: Equilibrium-Free Nonlinear Stability}

We first specialize Theorem~\ref{thm:invariant_tube} to autonomous operation with constant power references.
Here, a constant reference $x_c\in\mathcal D$ defines a fixed invariant tube around $x_c+\ker(R)$.

\begin{corollary}[Contraction Stability in Autonomous Regime]
\label{cor:autonomous_tube}
Consider the autonomous system $\dot x=f(x)$, semi-contracting on a convex domain $\mathcal D$ under $\|\cdot\|_R$ with rate $c>0$.
Let $x_c\in\mathcal D$ be a constant reference, and define the residual $r_{\mathrm{res}}:=\|f(x_c)\|_R $.
For any radius $\bar r\ge r_{\mathrm{res}}/c$, let
\begin{equation}
    \mathcal T_{\bar r}(x_c):=
    \left\{
    x\;\middle|\;
    \inf_{v\in\ker(R)}
    \|x-(x_c+v)\|_R \le \bar r
    \right\}.
\end{equation}
Then any $\mathcal T_{\bar r}(x_c)\subseteq\mathcal D$ is forward-invariant.
Moreover, it contains an equilibrium manifold $x^*+\ker(R)$ satisfying
\begin{equation}
    \inf_{v\in\ker(R)}
    \|x_c-(x^*+v)\|_R
    \le
    \frac{r_{\mathrm{res}}}{c}.
\end{equation}
\end{corollary}


Corollary~\ref{cor:autonomous_tube}, proved in Appendix~\ref{app:proof_cor_autonomous}, 
provides an equilibrium-free stability assessment for autonomous systems without disturbances.
Instead of solving for the exact equilibrium, the stability region can be certified from any physically meaningful reference $x_c$, 
ensuring bounded trajectories that converge to the equilibrium manifold.
When $x_c$ is an exact equilibrium, the residual vanishes,  
and the result reduces to the classical equilibrium-known ROA estimation, 
demonstrating the generalization of conventional stability analysis.


\vspace{-10pt}
\subsection{Non-Autonomous Operation: Tracking and Robustness Bounds}

We next consider non-autonomous operation under time-varying power injections.
As the operating condition continuously drifts, 
the stability objective shifts from asymptotic convergence to bounded projected deviations.
Building on Theorem~\ref{thm:invariant_tube}, 
we address two practical regimes: tracking of a moving quasi-steady manifold under slow variations, 
and robustness around a nominal manifold under fast perturbations.

\subsubsection{Slow-Varying Injections}

We first consider slowly varying injections that gradually shift the operating condition, 
such as renewable fluctuations and load ramps.
Motivated by frozen-equilibrium analysis~\cite{Ru2024Slowly}, 
we take the quasi-steady manifold induced by the instantaneous input as the reference.
The following result bounds the projected tracking error relative to this moving manifold.

\begin{corollary}[Quasi-Steady Tracking Bounds]\label{cor:slow_varying}

Consider $\dot{x}=f(x,u(t))$, semi-contracting on $\mathcal D$ under $\|\cdot\|_R$ with rate $c>0$.
Let $\mathcal U$ be an admissible input set on which there exists a continuously differentiable quasi-steady map
$x^*(u):\mathcal U\to\mathcal D$ satisfying
\begin{equation}
    \|f(x^*(u), u)\|_R = 0, \qquad \left\| \frac{\partial x^*}{\partial u}(u) \right\|_R \le H,
\end{equation}
for all $u\in\mathcal U$, where $H>0$ is a uniform sensitivity bound.

Under the slow-variation condition $\|\dot{u}(t)\| \le \epsilon$, 
the trajectory relative to
$x_c(t):=x^*(u(t))$ satisfies, for all $t\ge 0$,
\begin{multline} \label{eq:slow_tracking_transient}
 \inf_{v\in\ker(R)}\|x(t)-(x_c(t)+v)\|_R \\ \le 
 e^{-ct}\rho(0)+\frac{H\epsilon}{c}\bigl(1-e^{-ct}\bigr)
\end{multline}
where $\rho(0) \ge \inf_{v\in\ker(R)} \|x(0) - (x_c(0) + v)\|_R$. 
Consequently, the ultimate bound is given by
\begin{equation} \label{eq:slow_tracking_ultimate}
    \limsup_{t\to\infty} \inf_{v\in\ker(R)} \|x(t) - (x_c(t) + v)\|_R \le \frac{H\epsilon}{c}.
\end{equation}
\end{corollary}

Corollary~\ref{cor:slow_varying}, proved in Appendix~\ref{app:proof_cor_slow_varying}, 
shows that the tracking bound is shaped 
by the input variation rate $\epsilon$, the semi-contraction rate $c$, and the input sensitivity $H$.
Tracking improves with slower variations, stronger restoring dynamics, and lower input sensitivity.

The sensitivity $H$ serves as a critical indicator of stability margins.
By the implicit function theorem, $H$ scales with the inverse of the relevant power-flow Jacobian.
Near synchronization or voltage-stability limits, this Jacobian becomes ill-conditioned, increasing $H$ and degrading the bound $H\epsilon/c$ even when $c>0$.
The inflated bound therefore signals increased operating-point sensitivity and reduced stability margin.



\subsubsection{Fast Perturbations}
We next consider fast perturbations, such as rapid load variations, sensor noise, 
and intermittent renewable fluctuations. 
In this case, the goal is to establish a robustness bound on the deviation 
from the nominal stable manifold associated with a constant input.

\begin{corollary}[Robust Boundedness]
\label{cor:fast_robustness}
Consider $\dot{x}=f(x,u(t))$, semi-contracting on $\mathcal D$ under $\|\cdot\|_R$ with rate $c>0$.
For a nominal input $\bar u\in\mathcal U$, let $x^*(\bar u)\in\mathcal D$ satisfy
$\|f(x^*(\bar u),\bar u)\|_R=0$.
Consider a time-varying input $u(t)$ satisfying $\|u(t)-\bar u\|_2 \le \delta$ for all $t\ge 0$, 
and set $x_c(t):=x^*(\bar u)$.

Assume that the vector field is uniformly Lipschitz in the input at the nominal manifold, namely,
\begin{equation}
    \|f(x^*(\bar u),u)-f(x^*(\bar u),\bar u)\|_R
    \le
    L_u\|u-\bar u\|_2,
    \ \ \  \forall u\in\mathcal U.
\end{equation}
Then, for all $t \ge 0$, the deviation satisfies
\begin{multline} \label{eq:fast_robust_transient}
    \inf_{v\in\ker(R)} \|x(t) - (x_c(t) + v)\|_R  \\ \le 
    e^{-ct}\rho(0) + \frac{L_u\delta}{c}\big(1 - e^{-ct}\big).
\end{multline}
where $\rho(0) \ge \inf_{v\in\ker(R)} \|x(0) - (x_c(0) + v)\|_R$. 
Consequently, the ultimate bound is given by 
\begin{equation} \label{eq:fast_robust_ultimate}
    \limsup_{t\to\infty} \inf_{v\in\ker(R)} \|x(t) - (x_c(t) + v)\|_R \le \frac{L_u\delta}{c}.
\end{equation}
\end{corollary}

Corollary~\ref{cor:fast_robustness}, proved in Appendix~\ref{app:proof_cor_fast}, 
provides a robustness guarantee against fast injection perturbations.
Unlike the slow-varying case, the objective is not to track a drifting quasi-steady manifold, 
but to confine trajectories near the nominal manifold $x^*(\bar u)+\ker(R)$.
The ultimate bound $L_u\delta/c$ is determined by the perturbation magnitude $\delta$, 
the vector field sensitivity $L_u$, and the semi-contraction rate $c$.
Thus, robustness improves with weaker perturbations, lower sensitivity, and stronger restoring dynamics.






\begin{figure}[t]
\centering
\includegraphics[width=0.81\linewidth]{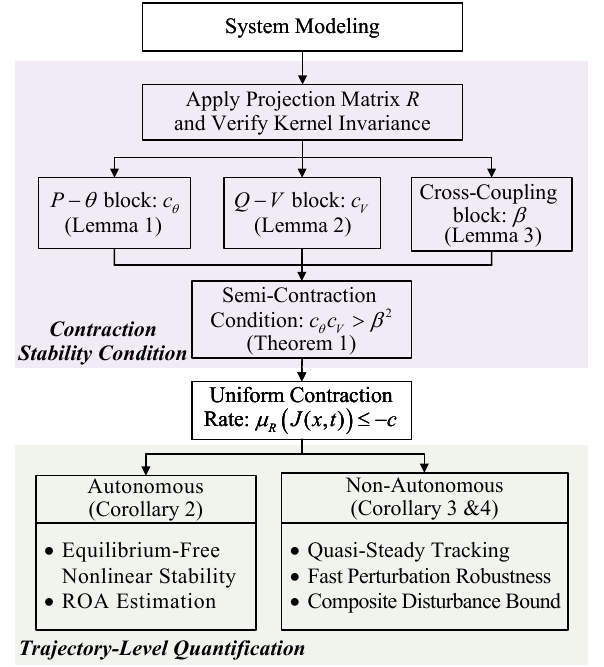}
\vspace{-5pt} 
\caption{Unified framework for contraction stability analysis.}
\label{fig:flow}
\vspace{-4pt}
\end{figure}

\begin{remark}[Unified Bounds for Composite Disturbances] \label{rem:unification} 
In practice, time-varying injections often combine slow variations and fast perturbations,
i.e., $u(t)=\bar u(t)+\tilde u(t)$, with
$\|\dot{\bar u}(t)\|_2\le \epsilon$ and $\|\tilde u(t)\|_2\le \delta$.

Considering $x^*(\bar u(t))$ as the moving reference, the residual contains additive slow- and fast-varying contributions, yielding
\begin{multline}
    \inf_{v\in\ker(R)}
    \|x(t)-(x^*(\bar u(t))+v)\|_R
    \\
    \le
    e^{-ct}\rho(0)
    +
    \frac{H\epsilon+L_u\delta}{c}\bigl(1-e^{-ct}\bigr),
\end{multline}
with the corresponding ultimate bound
\begin{equation*} 
    \limsup_{t\to\infty} \inf_{v\in\ker(R)} \|x(t)-(x^*(\bar{u}(t))+v)\|_R \le \frac{H\epsilon+L_u\delta}{c}.
\end{equation*}
Thus, the bound provides a unified certificate for projected deviations under realistic composite disturbances.
\end{remark} 

This section converts the contraction stability analysis into explicit trajectory-level certificates in the projected space.
It provides an equilibrium-free stability characterization for autonomous operation, 
together with tracking and robustness bounds for non-autonomous operation under slow and fast injection fluctuations.
The overall analytical framework is summarized in Fig.~\ref{fig:flow}.


\section{Case Studies}
The proposed method is validated on a modified IEEE 9-bus system implemented in MATLAB/Simulink. 
After Kron elimination of non-generator buses, the system is reduced to 3-bus equivalent network shown in Fig.~\ref{fig:3bus}. 
Despite its compact scale, the reduced system retains the key mechanisms considered in the analysis, 
including rotational symmetry, lossy network effects, heterogeneous droop settings, and coupled angle-voltage dynamics.

Each bus is connected to a droop-controlled converter and interconnected through equivalent lossy impedances $R_{ik}+jX_{ik}$.
All quantities are in per unit (p.u.) unless otherwise specified, 
with parameters listed in Table~\ref{tab:3bus_param}.
The admissible domain is set as $V_i\in[0.95,1.05]$ p.u. and $\gamma_{\max}=20^\circ$. 
The bounds computed over this domain yield $c_\theta$, $c_V$, and $\beta$, with $c_\theta c_V>\beta^2$. 
Thus, Theorem~\ref{thm:semi_contraction} certifies uniform semi-contraction on $\mathcal D$ with rate $c=0.184$.

\begin{figure}[t]
\vspace{2pt}
\centering
\includegraphics[width=0.7\linewidth]{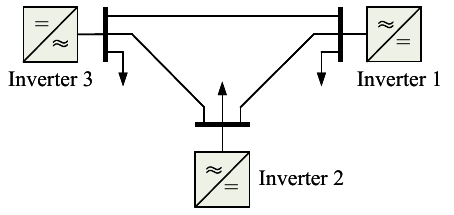}
\vspace{-8pt} 
\caption{Reduced 3-bus converter network derived from the modified IEEE 9-bus system.}
\label{fig:3bus}
\vspace{-16pt}
\end{figure}

\begin{table}[t]
\vspace{-6pt}
\centering
\caption{Parameters of the reduced 3-bus equivalent system.}
\vspace{-4pt}
\label{tab:3bus_param}
\begin{tabular}{c c c}
\toprule
\midrule
\textbf{Category} & \textbf{Parameter} & \textbf{Value} \\
\midrule

Base    &Apparent power base $S_b$ & 10~MVA \\
        &Voltage base $V_b$ (line RMS) &10~kV \\
        &Frequency base $f_b$ & 50~Hz \\
\addlinespace

Control &Active power droop $m_{p,i}$ & $[0.040,\,0.035,\,0.050]$   \\
        &Reactive power droop $n_{q,i}$ & $[0.020,\,0.020,\,0.015]$ \\
        &Voltage time constants $\tau_{v,i}$ & $[0.70,\,0.80,\,0.11]$~s\\
\addlinespace

Network & Equivalent impedance $Z_{12}$ & $0.08 + j0.11$ \\
        & Equivalent impedance $Z_{13}$ & $0.08 + j0.11$ \\
        & Equivalent impedance $Z_{23}$ & $0.10 + j0.12$ \\

\bottomrule
\end{tabular}
\end{table}

\vspace{-10pt}
\subsection{Autonomous Operation With Constant Power References}
\begin{figure}[t]
\centering
\includegraphics[width=0.92\linewidth]{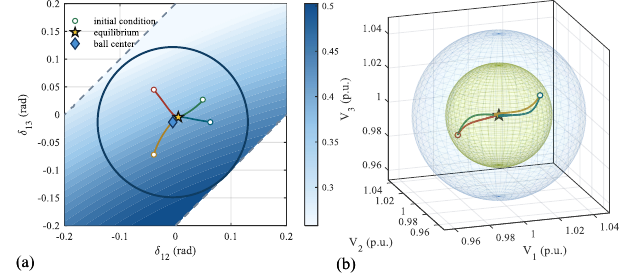}
\vspace{-6pt} 
\caption{Invariant-set verification for the autonomous case.  
(a) Invariant ball in the angle-difference subspace, constructed around a selected reference point (blue diamond).
(b) Invariant shells in the voltage subspace.}
\label{fig:auto_simu1}
\vspace{-12pt} 
\end{figure}

\begin{figure}[t]
\centering
\includegraphics[width=0.95\linewidth]{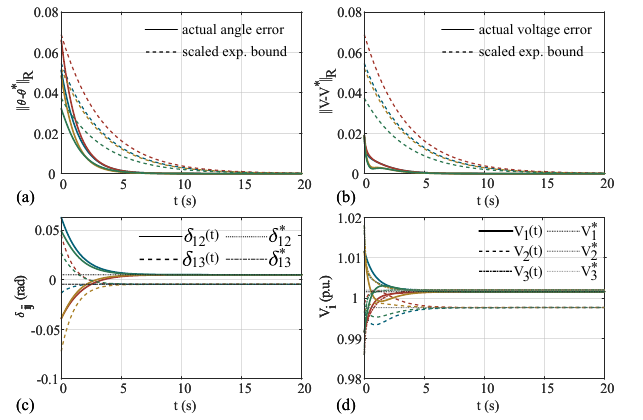}
\vspace{-6pt} 
\caption{Time-domain verification for the autonomous case. 
(a), (b) Projected angle and voltage errors for multiple initial conditions, 
with trajectories (solid) bounded by the exponential envelopes (dashed) with rate $c = 0.184$. 
(c), (d) Angle differences and voltage magnitudes converging to the stable manifold.}
\label{fig:auto_simu2}
\vspace{-4pt}
\end{figure}
To validate the proposed equilibrium-free analysis under autonomous operation, we apply constant power references 
$P_i^{\mathrm{ref}}=[0.25,\,0.20,\,0.25]$ and 
$Q_i^{\mathrm{ref}}=[0.10,\,-0.10,\,0.10]$. 
This case verifies the regional stability result in Corollary~\ref{cor:autonomous_tube} 
and the projected convergence property in Corollary~\ref{cor:trajectory}.

Fig.~\ref{fig:auto_simu1}(a) shows a two-dimensional slice of the invariant set in the angle-difference subspace. 
Following Corollary~\ref{cor:autonomous_tube}, the set is constructed around a selected reference point (blue diamond) by evaluating its residual, 
without prior knowledge of the exact equilibrium.
Similarly, Fig.~\ref{fig:auto_simu1}(b) shows the corresponding invariant shells in the voltage subspace. 
The outer blue shell represents the maximal forward-invariant set derived from theoretical bounds, 
whereas the inner yellow-green shell represents the minimal forward-invariant set determined by multiple initial conditions. 
All trajectories remain within the certified set and converge to the equilibrium (gold star).

Fig.~\ref{fig:auto_simu2} verifies the contraction behavior in the time domain. 
Quantified by $\|\theta-\theta^*\|_R$ and $\|V-V^*\|_R$, 
Figs.~\ref{fig:auto_simu2}(a) and (b) show that the trajectory distances to the final steady state 
strictly decay within the theoretical envelopes (dashed lines) with rate $e^{-0.184t}$. 
Figs.~\ref{fig:auto_simu2}(c) and (d) show convergence of angle differences and voltage magnitudes to the stable manifold from different initial states.

Collectively, these results validate the proposed equilibrium-free analysis in capturing 
both invariant-set geometry and convergence behavior under autonomous operation.

\begin{figure}[t]
\centering
\includegraphics[width=0.98\linewidth]{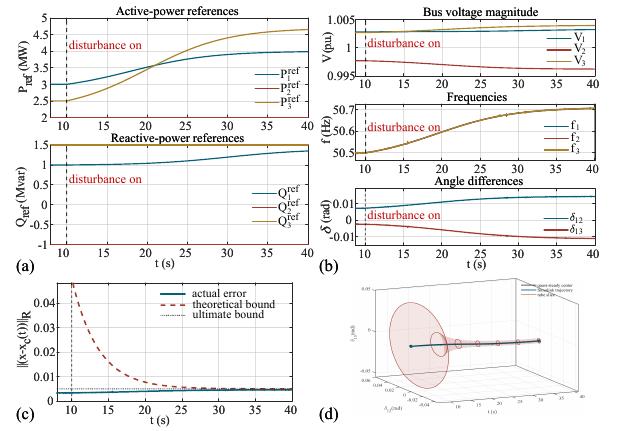}
\vspace{-6pt} 
\caption{Time-domain verification under slowly varying disturbances, with $\epsilon\approx0.014$. 
(a) Power references. 
(b) State responses after disturbance injection at $t=10\,\mathrm{s}$. 
(c) Projected actual error and theoretical bound. 
(d) Evolution of the invariant set in the angle-difference subspace.}
\label{fig:non_slow}
\vspace{-12pt}
\end{figure}
\begin{figure}[t]
\centering
\includegraphics[width=1\linewidth]{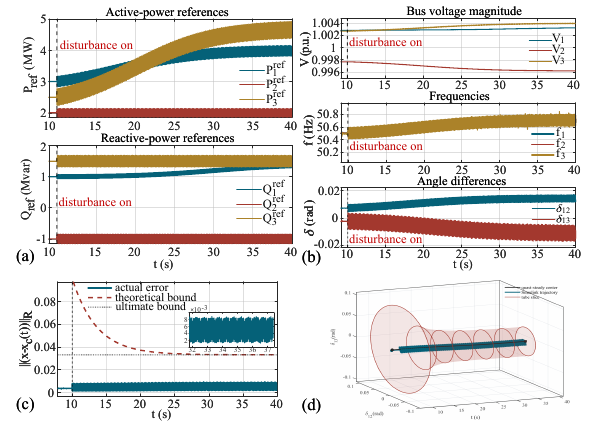}
\vspace{-23pt} 
\caption{Time-domain verification under composite disturbances, with $\epsilon\approx0.014$ and $\delta=0.04$.
(a) Power references. 
(b) State responses after disturbance injection at $t=10\,\mathrm{s}$. 
(c) Projected actual error and theoretical bound. 
(d) Evolution of the invariant set in the angle-difference subspace.}
\label{fig:non_mix}
\vspace{-10pt}
\end{figure}

\vspace{-8pt} 
\subsection{Non-Autonomous Operation Under Time-Varying Injections}
We next consider non-autonomous operation under time-varying power injections. 
The power references are set to 
$P_i^{\mathrm{ref}}=[0.30,\,0.20,\,0.25]$ and 
$Q_i^{\mathrm{ref}}=[0.10,\,-0.10,\,0.15]$ prior to the disturbance. 
Time-varying disturbances are applied at $t=10\,\mathrm{s}$.

\subsubsection{Slow-Varying Injections}
For the applied ramp disturbance, the corresponding parameters are 
$\epsilon=1.40\times10^{-2}$ and $H=0.059$. 
Figs.~\ref{fig:non_slow}(a) and (b) show the time-varying power references and the resulting state trajectories.
In Fig.~\ref{fig:non_slow}(c), the projected tracking error remains below the theoretical bound, 
verifying the enclosure predicted by Corollary~\ref{cor:slow_varying} around the moving quasi-steady manifold.
Fig.~\ref{fig:non_slow}(d) further illustrates the evolution of the corresponding invariant set in the angle-difference subspace.

\subsubsection{Composite Disturbances}
We next superimpose a fast perturbation on the slow-varying reference. 
For the fast component, the disturbance amplitude is set to $\delta=0.04$ relative to the reference power, 
with a corresponding Lipschitz constant of $L_u=0.132$.
The resulting bound becomes $(H\epsilon+L_u\delta)/c\approx3.31\times10^{-2}$. 
Specifically, Figs.~\ref{fig:non_mix}(a) and (b) show the composite power references and the resulting system responses, 
while Figs.~\ref{fig:non_mix}(c) and (d) validate the unified bound in Remark~\ref{rem:unification}.

Although more conservative in the composite case, the bound still provides a rigorous robustness certificate, 
validating the method under time-varying injections.

\section{Conclusion}
This paper proposed an equilibrium-free contraction stability framework for 
GFM converter-based microgrids with high renewable penetration.
By projecting out the rotational symmetry and bounding the coupled angle-voltage dynamics, 
we derived a computable semi-contraction condition over an admissible domain. 
The condition was further converted into forward-invariant certificates, yielding 
trajectory-level performance guarantees in the projected space.
For autonomous operation, the method provides 
an equilibrium-free nonlinear stability characterization that extends classical ROA estimation.
For non-autonomous operation, it establishes explicit tracking and robustness bounds under different power disturbances.
The case study validated the proposed stability results.

Future work will extend the proposed method to larger interconnected systems 
with synchronous generators and diverse converter models. 
Efficient online implementations will also be developed for stability assessment and control design.



\appendices
\section{Analytical Expressions of the Jacobian Blocks}
\label{app:jacobian_derivation}

Recall the droop dynamics
\begin{subequations}
\begin{align}
    \dot{\theta}_i &= \omega_{\mathrm{nom}} - m_{p,i} \big( P_i(x) - P_{i,0}^{\mathrm{ref}}(t) \big), \\
    \tau_{v,i} \dot{V}_i &= V_{\mathrm{nom},i} - V_i - n_{q,i} \big( Q_i(x) - Q_{i,0}^{\mathrm{ref}}(t) \big).
\end{align}
\end{subequations}
For brevity, define the pairwise AC power-flow coupling terms from node $i$ to node $k$ as:
\begin{subequations}
\begin{align}
    P_{ik} &:= V_i V_k (G_{ik} \cos\theta_{ik} + B_{ik} \sin\theta_{ik}), \\
    Q_{ik} &:= V_i V_k (G_{ik} \sin\theta_{ik} - B_{ik} \cos\theta_{ik}).
\end{align}
\end{subequations}
The resulting Jacobian subblocks $J_{\theta\theta}$, $J_{\theta V}$, $J_{V\theta}$, and $J_{VV}$ are computed as follows.

\subsubsection{Active Power--Angle Block ($J_{\theta\theta}$)}
For $i\neq k$,
\begin{align*}
    [J_{\theta\theta}]_{ik} &= -m_{p,i}Q_{ik},&
    [J_{\theta\theta}]_{ii} &= m_{p,i}\sum_{k\neq i}Q_{ik}.
\end{align*}

\subsubsection{Reactive Power--Voltage Block ($J_{VV}$)}
For $i\neq k$,
\begin{align*}
    [J_{VV}]_{ik} &= -\frac{n_{q,i}}{\tau_{v,i}}\frac{Q_{ik}}{V_k},&
    [J_{VV}]_{ii} &= -\frac{1}{\tau_{v,i}}
    -\frac{n_{q,i}}{\tau_{v,i}}\frac{Q_i-V_i^2B_{ii}}{V_i}.
\end{align*}

\subsubsection{Cross-Coupling Blocks ($J_{\theta V}$ and $J_{V\theta}$)}
For $i\neq k$,
\begin{equation*}
\begin{aligned}
[J_{\theta V}]_{ik}&=-\frac{m_{p,i}P_{ik}}{V_k},&
[J_{\theta V}]_{ii}&=-\frac{m_{p,i}(P_i+V_i^2G_{ii})}{V_i},\\
[J_{V\theta}]_{ik}&=\frac{n_{q,i}}{\tau_{v,i}}P_{ik},&
[J_{V\theta}]_{ii}&=-\frac{n_{q,i}}{\tau_{v,i}}\sum_{k\neq i}P_{ik}.
\end{aligned}
\end{equation*}

Hence, $\sum_k[J_{\theta\theta}]_{ik}=\sum_k[J_{V\theta}]_{ik}=0$, 
which gives $J(x,t)\mathbf{1}_{2N}^{\theta}=\mathbf{0}_{2N}$ and verifies kernel invariance.

\vspace{-6pt}
\section{Proofs of the Block Bounds}
\subsection{Proof of Lemma \ref{lem:bound_theta}}
\label{appendix:proof_theta_block}
\begin{proof}
Substituting~\eqref{eq:laplacian_decomposition} into the projected block yields
\begin{equation}
    \mathcal{S}_{\theta\theta}(x,t) = R_\theta \bigl( -L_{\mathrm{sym}}(x) + \Delta(x) \bigr) R_\theta^\top.
\end{equation}

By Weyl's inequality~\cite[Thm.~4.3.7]{horn2012matrix},
\begin{multline*}
    \lambda_{\max}\bigl(\mathcal{S}_{\theta\theta}(x,t)\bigr) \le  \\
    \lambda_{\max}\bigl(-R_\theta L_{\mathrm{sym}}(x) R_\theta^\top\bigr) + \lambda_{\max}\bigl(R_\theta \Delta(x) R_\theta^\top\bigr).
\end{multline*}

Since the rows of $R_\theta$ form an orthonormal basis of $\mathbf{1}_N^\perp$
and $L_{\mathrm{sym}}(x)\succeq \underline L$ on this subspace for all $x\in\mathcal D$, 
\begin{equation*}
    \lambda_{\max}\!\left(-R_\theta L_{\mathrm{sym}}(x)R_\theta^\top\right)
    \le -\lambda_2(\underline L).
\end{equation*}

For the second term, $\Delta(x)$ is diagonal, so its maximum eigenvalue is its largest diagonal element.
By eigenvalue interlacing property, for any $x\in\mathcal D$,
\begin{equation*}
    \lambda_{\max}\!\left(R_\theta\Delta(x)R_\theta^\top\right)
    \le \lambda_{\max}(\Delta(x))
    =\max_{i\in\mathcal V}\Delta_{ii}(x)
    \le\delta_\theta .
\end{equation*}
Combining the two bounds gives
\begin{equation}
    \lambda_{\max}(\mathcal S_{\theta\theta}(x,t))
    \le -\lambda_2(\underline L)+\delta_\theta
    =-c_\theta .
\end{equation}
This is negative under $\lambda_2(\underline L)>\delta_\theta$, completing the proof.
\end{proof}

\vspace{-10pt}
\subsection{Proof of Lemma~\ref{lem:bound_V}}
\label{appendix:proof_voltage_block}
\begin{proof}

From Appendix~\ref{app:jacobian_derivation}, the Gershgorin center of 
$\mathcal S_{VV}$ is
\begin{equation*}
    c_i(x)=-\frac{1}{\tau_{v,i}}
    +\frac{n_{q,i}}{\tau_{v,i}}V_iB_{ii}
    -\frac{n_{q,i}}{\tau_{v,i}}\frac{Q_i}{V_i}.
\end{equation*}
Expanding $Q_i$ gives the term 
$2n_{q,i}V_iB_{ii}/\tau_{v,i}$, whose algebraic maximum is attained at 
$V_i=\underline V$ since $B_{ii}<0$. Maximizing the remaining trigonometric terms over 
$\mathcal D$ yields $\bar c_i$.

For the off-diagonal entries, symmetrization gives
\begin{multline*}
    r_i(x)=\frac{1}{2}\sum_{k\neq i}\bigg|
    \left(\frac{n_{q,i}}{\tau_{v,i}}V_i+\frac{n_{q,k}}{\tau_{v,k}}V_k\right)
    B_{ik}\cos\theta_{ik} \\
    -
    \left(\frac{n_{q,i}}{\tau_{v,i}}V_i-\frac{n_{q,k}}{\tau_{v,k}}V_k\right)
    G_{ik}\sin\theta_{ik}
    \bigg|.
\end{multline*}
Bounding the voltage magnitudes by $\overline V$ and maximizing over 
$|\theta_{ik}|\le\gamma_{\max}$ yields the state-independent radius bound $\bar r_i$.

By Gershgorin's theorem~\cite[Thm.~6.1.1]{horn2012matrix},
\begin{equation}
    \lambda_{\max}\bigl(\mathcal S_{VV}(x,t)\bigr)
    \le \max_{i\in\mathcal V}(\bar c_i+\bar r_i)
    =-c_V .
\end{equation}
The assumption $\bar c_i+\bar r_i<0$ for all $i\in\mathcal V$ implies $c_V>0$, completing the proof.
\end{proof}

\vspace{-10pt}
\section{Proofs of the Invariant-Set Corollaries}
\subsection{Proof of Corollary~\ref{cor:autonomous_tube}}
\label{app:proof_cor_autonomous}
\begin{proof}

With the constant reference $x_c(t)\equiv x_c$, the residual is 
$r(t)\equiv r_{\mathrm{res}}$, and the comparison dynamics in 
Theorem~\ref{thm:invariant_tube} reduces to
$\dot{\rho}(t)=-c\rho(t)+r_{\mathrm{res}}$. Hence every radius
$\bar r\ge r_{\mathrm{res}}/c$ with 
$\mathcal T_{\bar r}(x_c)\subseteq\mathcal D$ defines a forward-invariant tube.

It remains to show that the invariant tube contains an equilibrium manifold. 
For radius $r_{\mathrm{res}}/c$, the projected set associated with 
$\mathcal T_{r_{\mathrm{res}}/c}(x_c)$ is compact, convex, and forward-invariant. 
By applying Brouwer's fixed-point theorem~\cite[Thm.~55.6]{munkres2000topology}, 
there exists at least one equilibrium $x^*$ in this set. 
By rotational symmetry, this gives the equilibrium manifold $x^*+\ker(R)$ such that
$\inf_{v\in\ker(R)}\|x_c-(x^*+v)\|_R\le r_{\mathrm{res}}/c$.

Finally, semi-contraction on $\mathcal D$ implies strict decay of transverse distances between any two trajectories, 
which precludes multiple distinct equilibrium manifolds.
\end{proof}

\vspace{-10pt}
\subsection{Proof of Corollary~\ref{cor:slow_varying}}
\label{app:proof_cor_slow_varying}
\begin{proof}
Since $\|f(x_c(t), u(t))\|_R = 0$, the residual in Theorem~\ref{thm:invariant_tube} simplifies to 
\begin{equation*}
    r(t)=\|f(x_c(t),u(t))-\dot x_c(t)\|_R=\|\dot x_c(t)\|_R.
\end{equation*} 
By the chain rule, 
$\dot x_c(t)=\frac{\partial x^*}{\partial u}(u(t))\dot u(t)$. 
Using norm submultiplicativity, the residual satisfies
\begin{equation*}
    r(t)\le 
    \left\|\frac{\partial x^*}{\partial u}(u(t))\right\|_R
    \|\dot u(t)\|_2
    \le H\epsilon .
\end{equation*}
Solving the scalar comparison dynamics $\dot{\rho}=-c\rho+r(t)$ gives
\begin{equation} \label{eq:convolu}
    \rho(t)=e^{-c(t-t_0)}\rho(t_0)
    +\int_{t_0}^{t}e^{-c(t-s)}r(s)\,ds .
\end{equation}
Setting $t_0=0$ and using $r(s)\le H\epsilon$ in~\eqref{eq:convolu} yields the claimed transient estimate, 
and the ultimate bound follows by letting $t\to\infty$.
\end{proof}

\vspace{-10pt}
\subsection{Proof of Corollary \ref{cor:fast_robustness}}
\label{app:proof_cor_fast}
\begin{proof}
For the constant reference $x_c(t)\equiv x^*(\bar u)$, we have 
$\dot x_c(t)=0$ and $\|f(x_c(t),\bar u)\|_R=0$. 
Thus, the residual defined in Theorem~\ref{thm:invariant_tube} simplifies to
\begin{equation*}
    r(t)=\|f(x_c(t),u(t))-\dot x_c(t)\|_R
    =\|f(x_c(t),u(t))\|_R .
\end{equation*}
By the uniform Lipschitz condition, 
$r(t)\le L_u\|u(t)-\bar u\|_2\le L_u\delta$.
Solving the comparison dynamics $\dot\rho=-c\rho+r(t)$ with this bound gives the stated transient and ultimate bounds.
\end{proof}


\bibliography{reference}

@article{Willems1974,
  author  = {J. L. Willems},
  title   = {A Partial Stability Approach to the Problem of Transient Power System Stability},
  journal = {International Journal of Control},
  year    = {1974},
  volume  = {19},
  number  = {1},
  pages   = {1-14},
  doi     = {10.1080/00207177408932606}
}

@book{Chiang2011,
  author    = {Hsiao-Dong Chiang},
  title     = {Direct Methods for Stability Analysis of Electric Power Systems: Theoretical Foundation, BCU Methodologies, and Applications},
  edition   = {2nd},
  publisher = {Wiley},
  address   = {Hoboken, NJ, USA},
  year      = {2011}
}

@misc{lin2020roadmap,
  author       = {Lin, Yashen and others},
  title        = {Research Roadmap on Grid-Forming Inverters},
  howpublished = {National Renewable Energy Laboratory (NREL), Technical Report NREL/TP-5D00-73476},
  year         = {2020}
}

@book{Kundur1994,
  author    = {Prabha Kundur},
  title     = {Power System Stability and Control},
  publisher = {McGraw-Hill},
  address   = {New York, NY, USA},
  year      = {1994}
}

@article{dorfler2012kron,
  title={Kron reduction of graphs with applications to electrical networks},
  author={D{\"o}rfler, Florian and Bullo, Francesco},
  journal={IEEE Transactions on Circuits and Systems I: Regular Papers},
  volume={60},
  number={1},
  pages={150-163},
  year={2012},
  publisher={IEEE}
}

@ARTICLE{yang2024augment,
  author={Yang, Peng and Liu, Feng and Liu, Tao and Hill, David J.},
  journal={IEEE Transactions on Automatic Control}, 
  title={Augmented Synchronization of Power Systems}, 
  year={2024},
  volume={69},
  number={6},
  pages={3673-3688},
  publisher={IEEE}
  }

@INPROCEEDINGS{wang2023droop,
  author={Wang, Zhenhua and Li, Hepeng and He, Haibo and Sun, Yan},
  booktitle={2023 IEEE Power \& Energy Society Innovative Smart Grid Technologies Conference (ISGT)}, 
  title={Robust Distributed Finite-time Secondary Frequency Control of Islanded AC Microgrids With Event-Triggered Mechanism}, 
  year={2023},
  volume={},
  number={},
  pages={1-5},
  publisher={IEEE}
  }

@article{Schiffer2014DroopStability,
  author  = {Johannes Schiffer and Romeo Ortega and Alessandro Astolfi and J{\"o}rg Raisch and Tevfik Sezi},
  title   = {Conditions for stability of droop-controlled inverter-based microgrids},
  journal = {Automatica},
  volume  = {50},
  number  = {10},
  pages   = {2457-2469},
  year    = {2014},
  month   = {10},
  publisher={Elsevier}
}

@INPROCEEDINGS{Ru2024Slowly,
  author={Ru, Xi and Yang, Peng and Liu, Ruitong and Jia, Chen and Duan, Fangwei and Liu, Feng},
  booktitle={2024 IEEE 7th Student Conference on Electric Machines and Systems (SCEMS)}, 
  title={Analysis of Power System Stability Under Slowly Varying Perturbations in Renewable Generation}, 
  year={2024},
  volume={},
  number={},
  pages={1-7},
  publisher={IEEE}}

@article{LOHMILLER1998,
title = {On Contraction Analysis for Non-linear Systems},
journal = {Automatica},
volume = {34},
number = {6},
pages = {683-696},
year = {1998},
issn = {0005-1098},
publisher={Elsevier},
author = {WINFRIED LOHMILLER and JEAN-JACQUES E. SLOTINE}
}

@ARTICLE{Jafarpour2022,
  author={Jafarpour, Saber and Cisneros-Velarde, Pedro and Bullo, Francesco},
  journal={IEEE Transactions on Automatic Control}, 
  title={Weak and Semi-Contraction for Network Systems and Diffusively Coupled Oscillators}, 
  year={2022},
  volume={67},
  number={3},
  pages={1285-1300},
  doi={10.1109/TAC.2021.3073096}}

@Book{FBbook,
  author =    {F. Bullo},
  title =     {Contraction Theory for Dynamical Systems},
  year =      2026,
  edition =   {{1.3}},
  publisher = {Kindle Direct Publishing},
  ISBN =      {979-8836646806},
  url =       {https://fbullo.github.io/ctds},
}

@book{horn2012matrix,
  title={Matrix Analysis},
  author={Horn, Roger A and Johnson, Charles R},
  edition={2},
  year={2012},
  publisher={Cambridge University Press}
}

@ARTICLE{Mohandes2019review,
  author={Mohandes, Baraa and Moursi, Mohamed Shawky El and Hatziargyriou, Nikos and Khatib, Sameh El},
  journal={IEEE Transactions on Power Systems}, 
  title={A Review of Power System Flexibility With High Penetration of Renewables}, 
  year={2019},
  volume={34},
  number={4},
  pages={3140-3155},
  publisher={IEEE}
  }

@INPROCEEDINGS{Milano2018,
  author={Milano, Federico and Dörfler, Florian and Hug, Gabriela and Hill, David J. and Verbič, Gregor},
  booktitle={2018 Power Systems Computation Conference (PSCC)}, 
  title={Foundations and Challenges of Low-Inertia Systems (Invited Paper)}, 
  year={2018},
  volume={},
  number={},
  pages={1-25},
  doi={10.23919/PSCC.2018.8450880}}

@article{Anvari_2016,
    doi = {10.1088/1367-2630/18/6/063027},
    year = {2016},
    month = {jun},
    publisher = {IOP Publishing},
    volume = {18},
    number = {6},
    pages = {063027},
    author = {Anvari, M and Lohmann, G and Wächter, M and Milan, P and Lorenz, E and Heinemann, D and Tabar, M Reza Rahimi and Peinke, Joachim},
    title = {Short term fluctuations of wind and solar power systems},
    journal = {New Journal of Physics}
}

@INPROCEEDINGS{Hill2006,
  author={Hill, D.J. and Guanrong Chen},
  booktitle={2006 IEEE International Symposium on Circuits and Systems (ISCAS)}, 
  title={Power systems as dynamic networks}, 
  year={2006},
  volume={},
  number={},
  pages={4 pp.-725},
  doi={10.1109/ISCAS.2006.1692687}}

@article{Florian2012,
author = {D\"{o}rfler, Florian and Bullo, Francesco},
title = {Synchronization and Transient Stability in Power Networks and Nonuniform Kuramoto Oscillators},
journal = {SIAM Journal on Control and Optimization},
volume = {50},
number = {3},
pages = {1616-1642},
year = {2012},
doi = {10.1137/110851584}}

@book{bhatia2002stability,
  title={Stability theory of dynamical systems},
  author={Bhatia, Nam P and Szeg{\"o}, Giorgio P},
  year={2002},
  publisher={Springer-Verlag},
  address={Berlin, Heidelberg}
}

@article{vorotnikov2005partial,
  title={Partial stability and control: The state-of-the-art and development prospects},
  author={Vorotnikov, V. I.},
  journal={Automation and Remote Control},
  volume={66},
  pages={511-561},
  year={2005},
  doi={10.1007/s10513-005-0099-9}
}

@article{dorfler2013synchronization,
  title={Synchronization in complex oscillator networks and smart grids},
  author={D{\"o}rfler, Florian and Chertkov, Michael and Bullo, Francesco},
  journal={Proceedings of the National Academy of Sciences},
  volume={110},
  number={6},
  pages={2005-2010},
  year={2013},
  publisher={National Acad Sciences}
}

@article{simpson2013synchronization,
  title={Synchronization and power sharing for droop-controlled inverters in islanded microgrids},
  author={Simpson-Porco, John W and D{\"o}rfler, Florian and Bullo, Francesco},
  journal={Automatica},
  volume={49},
  number={9},
  pages={2603--2611},
  year={2013},
  publisher={Elsevier}
}

@article{Zhu2018,
author = {Zhu, Lijun and Hill, David J.},
title = {Stability Analysis of Power Systems: A Network Synchronization Perspective},
journal = {SIAM Journal on Control and Optimization},
volume = {56},
number = {3},
pages = {1640-1664},
year = {2018},
doi = {10.1137/17M1118646}}

@article{2020microgrid,
  title={Microgrid stability definitions, analysis, and examples},
  author={Farrokhabadi, Mostafa and Ca{\~n}izares, Claudio A and Simpson-Porco, John W and Nasr, Ehsan and Fan, Lingling and Mendoza-Araya, Patricio A and others},
  journal={IEEE Transactions on Power Systems},
  volume={35},
  number={1},
  pages={13-29},
  year={2020},
  publisher={IEEE}
}

@ARTICLE{Vorobev2018,
  author={Vorobev, Petr and Huang, Po-Hsu and Al Hosani, Mohamed and Kirtley, James L. and Turitsyn, Konstantin},
  journal={IEEE Transactions on Power Systems}, 
  title={High-Fidelity Model Order Reduction for Microgrids Stability Assessment}, 
  year={2018},
  volume={33},
  number={1},
  pages={874-887},
  doi={10.1109/TPWRS.2017.2707400}}

@book{khalil2002nonlinear,
  title={Nonlinear Systems},
  author={Khalil, Hassan K},
  edition={3rd},
  year={2002},
  publisher={Prentice Hall},
  address={Upper Saddle River, NJ}
}

@book{munkres2000topology,
  title={Topology},
  author={Munkres, James R.},
  year={2000},
  publisher={Prentice Hall}
}

@ARTICLE{He2021,
  author={He, Xiuqiang and Geng, Hua},
  journal={IEEE Transactions on Power Systems}, 
  title={Transient Stability of Power Systems Integrated With Inverter-Based Generation}, 
  year={2021},
  volume={36},
  number={1},
  pages={553-556}
}
\bibliographystyle{IEEEtran}
\vfill

\vspace{11pt}
\end{document}